\documentclass{article}

\usepackage{microtype}
\usepackage{graphicx}
\usepackage{subcaption}
\usepackage{booktabs} 
\usepackage{stackrel}
\usepackage{derivative} 

\usepackage{hyperref}

\usepackage[accepted]{icml2023}

\usepackage{amsmath}
\usepackage{amssymb}
\usepackage{mathtools}
\usepackage{amsthm}

\usepackage[capitalize,noabbrev]{cleveref}

\theoremstyle{plain}
\newtheorem{theorem}{Theorem}[section]
\newtheorem{proposition}[theorem]{Proposition}
\newtheorem{lemma}[theorem]{Lemma}
\newtheorem{corollary}[theorem]{Corollary}
\theoremstyle{definition}
\newtheorem{definition}[theorem]{Definition}

\theoremstyle{remark}
\newtheorem{remark}[theorem]{Remark}

\newcommand{\game}{\mathcal{G}}
\newcommand{\agentset}{\mathcal{N}}
\newcommand{\edgeset}{\mathcal{E}}
\newcommand{\actionset}[1]{\mathcal{A}_{#1}}

\newcommand{\NE}{\mathbf{\bar{x}}}

\newcommand{\A}{\mathbf{A}}

\newcommand{\defeq}{\vcentcolon=}

\newcommand{\x}{\mathbf{x}}

\newcommand{\y}{\mathbf{y}}

\newcommand{\zeros}{\mathbf{0}}

\newcommand{\epsX}{\epsilon_X}
\newcommand{\epsY}{\epsilon_Y}

\newcommand{\R}{\mathbb{R}}

\newcommand{\X}{\mathcal{X}}

\usepackage[textsize=tiny]{todonotes}

\icmltitlerunning{Stability of Multi-Agent Learning: Convergence in Network Games with Many Players}

\begin{document}

\twocolumn[
\icmltitle{Stability of Multi-Agent Learning:\\ Convergence in Network Games with Many Players}



\icmlsetsymbol{equal}{*}

\begin{icmlauthorlist}
\icmlauthor{Aamal Hussain}{yyy,equal}
\icmlauthor{Dan Leonte}{yyy,equal}
\icmlauthor{Francesco Belardinelli}{yyy}
\icmlauthor{Georgios Piliouras}{comp,sch}
\end{icmlauthorlist}

\icmlaffiliation{yyy}{Imperial College London}
\icmlaffiliation{comp}{Singapore University of Technology and Design}
\icmlaffiliation{sch}{DeepMind}

\icmlcorrespondingauthor{Aamal Hussain}{aamal.hussain15@imperial.ac.uk}

\icmlkeywords{Multi-Agent Learning  \and Learning in Games \and Reinforcement Learning}

\vskip 0.3in
]



\printAffiliationsAndNotice{\icmlEqualContribution} 

\begin{abstract}
The behaviour of multi-agent learning in many player games has been shown to display complex dynamics outside of restrictive examples such as network zero-sum games. In addition, it has been shown that convergent behaviour is less likely to occur as the number of players increase. To make progress in resolving this problem, we study Q-Learning dynamics and determine a sufficient condition for the dynamics to converge to a unique equilibrium in any network game. We find that this condition depends on the nature of pairwise interactions and on the network structure, but is explicitly independent of the total number of agents in the game. We evaluate this result on a number of representative network games and show that, under suitable network conditions, stable learning dynamics can be achieved with an arbitrary number of agents.
\end{abstract}

\section{Introduction}

Determining the convergence of multi-agent learning is arguably amongst the most studied problems in game theory and online learning \cite{anagnostides:last-iterate,bai:optimistic,ewerhart:fp}. In this setting, agents are required to explore their state space to determine optimal actions, whilst simultaneously aiming to maximise their expected reward. To do this, each agent must react to the changing behaviour of the other agents so that, from the perspective of any given agent, the environment is non stationary. A large body of recent work has shown that this non stationarity leads to complex behaviours being displayed by learning dynamics in games with many agents \cite{sato:rps,piliouras:arbitrarilycomplex}. In addition, even in games where convergent behaviour can be achieved, it can be to one of multiple equilibria \cite{piliouras:poincare,villatoro:convention,sanders:chaos}. In fact, recent work suggests that as the number of agents in the system increase the likelihood for chaotic dynamics increases. This presents a significant challenge for making predictions in multi-agent settings.

Nevertheless, multi-agent learning has been the major driver of a number of successes in Artificial Intelligence and Machine Learning. Such examples include strong performance in competitive games \cite{,brown:poker,tuyls:stratego}, resource allocation \cite{grammatico:NAG,amelina:load-balancing} and robotics \cite{hamann:swarm,hernandes:patrolling}. Due to these successes, and the increasing use of multi-agent systems, it becomes important to develop a theoretical understanding of learning algorithms in various settings.

To this end, a number of advances have been made which consider learning in multi-agent settings. Indeed, strong positive results on convergence have been found in cooperative settings, such as \emph{potential games} \cite{candogan:nearpotential,harris:fp} and competitive settings, such as zero-sum network games \cite{piliouras:zerosum,kadan:exponential,cai:minimax}. However, a general framework for understanding learning behaviour must extend beyond these settings. In \cite{hussain:aamas} a strong positive result on convergence was found which states that an equilibrium can be reached in any game, given sufficient exploration by all agents. Unfortunately, the condition requires that the amount of exploration increases with the number of agents. To make matters worse, \cite{sanders:chaos} show that, as the number of agents increase in the game, learning dynamics are more likely to display chaotic behaviour. 

However, both of these works did not impose any structure on the interaction between agents. Rather, \cite{sanders:chaos} assume that all agents interact with all others. In reality, agents are more likely to interact according to an underlying communication network. In economic settings this may correspond to social interactions between agents, whilst in machine learning settings, networks are often used to enforce an underlying structure to the model. 

\paragraph{Model and Contributions} In light of this, we study learning in \emph{network games}, where interactions between agents can be constrained. On this model, we study the \emph{Q-Learning} dynamic \cite{sato:qlearning,tuyls:qlearning}, a well studied learning dynamic captures the balance between agents who explore their state space whilst maximising their reward.

Our main result tightens the requirement of sufficient exploration found by \cite{hussain:aamas} to achieve convergence to a unique equilibrium in any network game. In particular we find that the amount of exploration depends on the nature of the interaction between agents and, more importantly, the structure of the network. We examine how our bound explicitly depends on the total number of agents in the system and find that, for certain networks, there is no explicit dependence. This enables a higher number of agents to be introduced in the system without compromising stability. In addition, our result applies to all network games, and not only network zero sum games. In fact, we show how our results relate to existing statements in the literature. Finally, we validate our findings on a number of representative classes of games and networks.


\paragraph{Related Work}

The theory of evolutionary game dynamics models multi-agent interactions in which agents improve their actions through \emph{online learning} \cite{shalev:online}. The premise is that popular learning algorithms such as \emph{Hedge} \cite{krichene:hedge}, online gradient descent \cite{kadan:exponential} and Q-Learning \cite{sutton:barto,schwartz:MARL} can be approximated in continuous time by a dynamical system \cite{mertikopoulos:reinforcement,krichene:thesis,tuyls:qlearning}. This enables tools from the study of dynamical systems to be used to analyse the behaviour of the learning algorithm. This approach has yielded a number of successes, most notably in \emph{potential games} \cite{piliouras:potential,candogan:nearpotential,Monderer1996PotentialGames}, which model multi-agent cooperation, and \emph{network zero sum games} \cite{cai:minimax,abernethy:hamiltonian}, which models competition. In these settings, it is known that a number of learning dynamics converge to an equilibrium \cite{kadan:exponential,ewerhart:fp,piliouras:zerosum}.

Outside of these classes, the behaviour of learning is less certain \cite{anagnostides:last-iterate}. In particular, it is known that learning dynamics can exhibit complex behaviours such as cycles \cite{piliouras:cycles,imhof:cycles,pangallo:bestreply,shapley:twoperson} and chaos \cite{svs:chaos,chakraborty:chaos,sato:rps,pangallo:taxonomy}. Indeed, \cite{galla:complex} showed that the Experience
Weighted Attraction (EWA) dynamic, which is closely related to Q-Learning \cite{piliouras:zerosum} achieves
chaos in classes of two-player games. Advancing this result,
\cite{sanders:chaos} showed that chaotic dynamics become more prevalent as the number of agents increase, regardless of the exploration rates. Similar to the work in this paper, \cite{hussain:aamas} determine a sufficient condition on the exploration rates for Q-Learning to converge in any game, yet they also find that this condition increases with the number of agents. This presents a strong barrier in placing guarantees on the behaviour in multi-agent systems with many agents, outside of restrictive settings.

Our work also employs a number of tools from the study of variational inequalities in game theory. This is a well studied framework for analysing the structure of equilibrium sets in a game \cite{melo:network,facchinei:VI} and for studying the convergence of algorithms equilibrium seeking algorithms \cite{tatarenko:monotone,mertikopoulos:finite,mertikopoulos:concave,sorin:composite}. Recent advances in this field begin to consider the properties of network games. Notably, \cite{parise:network,melo:network} determine conditions under which the Nash Equilibrium of a network game is unique, and how these relate to properties of the network. Similarly, \cite{melo:qre} shows the uniqueness of various formulations of the Quantal Response Equilibrium (QRE) under particular choices of payoff functions. Whilst our results use similar techniques, we do not make such assumptions on the nature of the payoffs, but rather parameterise our final condition on the nature of interactions between agents. In addition, we consider the stability of learning.

In our work, we aim to address the problem of convergence in many-agent systems by considering games which are played on a network \cite{cai:minimax}. Extending the work of \cite{hussain:aamas}, we are able to find a sufficient condition on exploration rates so that the Q-Learning dynamics converge to a unique equilibrium. Importantly, we show that this is independent of the total number of agents in the system. To our knowledge this is the first work which shows the convergence of Q-Learning in arbitrary network games.

\section{Preliminaries} \label{sec::Prelims}
    
We begin in Section \ref{sec::model} by defining the network game model, which is the setting on which we study the Q-Learning dynamics, which
we describe in Section \ref{sec::LearningModel}.

\subsection{Game Model}\label{sec::model}

In this work, we consider \textit{network polymatrix games} \cite{cai:minimax}. A Network Game is described by the tuple $\game = (\agentset, \edgeset, (u_k, \actionset{k})_{k \in
\agentset})$, where $\agentset$ denotes a finite set of players $\agentset$ indexed by $k = 1, \ldots, N$. Each agent can choose from a finite set of actions $\actionset{k}$ indexed by $i = 1, \ldots, n$. We denote the
\emph {strategy} $\x_k$ of an agent $k$ as the probabilities with which they play their actions. Then, the set of all strategies of agent $k$ is $\Delta(\actionset{k}) := \left\{ \x_k
\in \R^{n} \, : \, \sum_i x_{ki} = 1,\, x_ {ki} \geq 0 \right\}$. Each agent is also given a payoff function 
$u_k \, : \Delta(\actionset{k}) \times \Delta(\actionset{-k}) \rightarrow \R$ where $\actionset{-k}$ denotes the action set
of all agents other than $k$.
Agents are connected via an underlying network defined by $\edgeset$. In particular, 
$\edgeset$ consists of pairs $(k, l) \in \agentset \times \agentset$ of connected agents $k$ and $l$. An equivalent way to define the network is through an \emph{adjacency matrix} $G$ so that
\begin{equation*}
    [G]_{k,l} = \begin{cases}
        1, \text{ if agents $k, l$ are connected} \\
        0, \text{ otherwise}
    \end{cases}.
\end{equation*}
It is assumed that the network is undirected, so that $G$ is a symmetric matrix. Each edge $(k, l) \in \edgeset$ corresponds to a pair of payoff matrices
$A^{kl}$, $A^{lk}$. With these specifications, the payoff received by each agent $k$ is given by
\begin{equation} \label{eqn::GPGPayoffs} u_k(\x_k, \x_{-k}) = \sum_{(k, l) \in \edgeset} \x_k \cdot A^{kl} \x_l.
\end{equation}
 For any $\x \in \Delta =: \times_k \Delta(\actionset{k})$, we can define the reward to agent $k$ for playing action $i$ as $r_{ki}(\x_{-k}) = \frac{\partial u_{ki}(\x)}{\partial x_{ki}}$. Under this notation, $u_k(\x_k, \x_{-k}) = \langle \x_k, r_k(\x) \rangle$. With this in place, we can define an equilibrium solution for the game.

\begin{definition}[Quantal Response Equilibrium (QRE)] A joint mixed strategy $\NE \in \Delta$ is a \emph
{Quantal Response Equilibrium} (QRE) if, for all agents $k$ and all actions $i \in \actionset{k}$
    \begin{equation*}
        \NE_{ki} = \frac{\exp(r_{ki}(\NE_{-k})/T_k)}{\sum_{j \in \actionset{k}} \exp(r_{kj}(\NE_{-k})/T_k)}.
    \end{equation*}
\end{definition}

The QRE \cite{camerer:bgt} is the prototypical extension of the Nash Equilibrium to the case of agents with bounded
rationality, parameterised by the \emph{exploration rate} $T_k$. In particular, the limit $T_k \rightarrow 0$
corresponds exactly to the Nash Equilibrium, whereas the limit $T_k \rightarrow \infty$ corresponds to a purely
irrational case, where action $i \in \actionset{k}$ is played with the same probability regardless of its associated reward. The link between the QRE and the Nash Equilibrium is made stronger through the following result.

\begin{proposition}[\cite{melo:qre}]
    Consider a game $\game = (\agentset, \edgeset, (u_k, \actionset{k})_{k \in
\agentset})$ and let $T_1, \ldots, T_N > 0$ be exploration rates. Define the perturbed game $\game^H = (\agentset, \edgeset, (u_k^H, \actionset{k})_{k \in
\agentset})$ with the payoff functions
\begin{equation*}
    u_k^H(\x_k, \x_{-k}) = u_k(\x_k, \x_{-k}) - T_k \langle \x_k, \ln \x_{k} \rangle.
\end{equation*}
Then $\NE \in \Delta$ is a QRE of $\game$ if and only if it is a Nash Equilibrium of $\game^H$.
\end{proposition}

\subsection{Learning Model} \label{sec::LearningModel}

 In this work, we analyse the \emph{Q-Learning dynamic}, a prototypical model for determining optimal policies by balancing exploration and exploitation. In this model, each
    agent $k \in \agentset$ maintains a history of the past performance of each of their actions. This history is updated
    via the Q-update
    \begin{equation*}
        Q_{ki}(\tau + 1) = (1 - \alpha_k) Q_{ki}(\tau) + \alpha_k r_{ki}(\x_{-k}(\tau)),
    \end{equation*}
    where $\tau$ denotes the current time step. $Q_{ki}(\tau)$ denotes the \emph{Q-value} maintained by agent $k$ about the
    performance of action $i \in S_k$. In effect $Q_{ki}$ gives a discounted history of the rewards received when $i$ is
    played, with $1 - \alpha_k$ as the discount factor.
    
    Given these Q-values, each agent updates their mixed strategies according to the Boltzmann distribution, given by
    \begin{equation*}
        x_{ki}(\tau) = \frac{\exp(Q_{ki}(\tau)/T_k) }{\sum_j \exp(Q_{kj}(\tau)/T_k)},
    \end{equation*}
    in which $T_k \in [0, \infty)$ is the \emph{exploration rate} of agent $k$.
    
    It was shown in \cite{tuyls:qlearning,sato:qlearning} that a continuous time approximation of the Q-Learning algorithm
    could be written as
    \begin{equation} \tag{QLD} \label{eqn::QLD}
        \frac{\dot{x}_{k i}}{x_{k i}}=r_{k i}\left(\mathbf{x}_{-k}\right)-\langle \mathbf{x}_k, r_k(\mathbf{x}) \rangle +T_k \sum_{j \in S_k} x_{k j} \ln \frac{x_{k j}}{x_{k i}},
    \end{equation}
    which we call the \emph{Q-Learning dynamics} (QLD). The fixed points of this dynamic coincide with the QRE of the game \cite{piliouras:zerosum}.

\subsection{Variational Inequalities and Game Theory}

Our aim in this work is to analyse the Q-Learning dynamics in network games without invoking any particular structure on the payoffs (e.g.~zero-sum). To do this, we employ the \emph{Variational Inequality} approach, which has been successfully applied towards the analysis of network games \cite{melo:network,parise:network,xu:networks-conflict} as well as learning in games \cite{mertikopoulos:finite,sorin:composite,hussain:aamas}. In this paper, we connect these areas of literature.

\begin{definition}[Variational Inequality] 
     Consider a set $\X \subset \R^d$ and a map $F \, : \X \rightarrow \R^d$. The Variational Inequality (VI) problem $VI(\X, F)$ is given as
        \begin{equation}\label{eqn::VIdef}
            \langle \x - \NE, F(\NE) \rangle \geq 0, \hspace{0.5cm} \text{ for all } \x \in \X.
        \end{equation}
    We say that $\NE \in \X$ belongs to the set of solutions to a variational inequality problem $VI(\X, F)$ if it satisfies (\ref{eqn::VIdef}).
\end{definition}

The premise of the variational approach to game theory \cite{facchinei:VI, Rosen1965ExistenceGames} is that the problem of finding equilibria of games can be reformulated as determining the set of solutions to a VI problem. This is done by choosing associating the set $\X$ with $\Delta$ and the map $F$ with the \emph{pseudo-gradient} of the game.

\begin{definition}[Pseudo-Gradient Map]
    The pseudo-gradient map of a game $\game = (\agentset, \edgeset, (u_k, \actionset{k})_{k \in
\agentset})$ is given by
    $F(\x) = (F_k(\x))_{k \in \agentset} = (-D_{\x_k} u_k(\x_k, \x_{-k}))_{k \in \agentset}$.
\end{definition}
The advantage of this formulation is that we can apply results from the study of Variational Inequalities to determine properties of the game. These results rely solely on the form of the pseudo-gradient map and so can generalise results which assume a potential or zero-sum structure of the game \cite{hussain:aamas,kadan:exponential}. 
\begin{lemma}[\cite{melo:qre}]
    Consider a game $\game = (\agentset, \edgeset, (u_k, \actionset{k})_{k \in
    \agentset})$ and for any $T_1, \ldots, T_N > 0$, let $F$ be the pseudo-gradient map of $\game^H$. Then $\NE \in \Delta$ is a QRE of $\game$ if and only if $\NE$ is a solution to $VI(\Delta, F)$.
\end{lemma}
With this correspondence in place, we can analyse properties of the pseudo-gradient map and its relation to properties of the game and the learning dynamic. One important property is \emph{monotonicity}.
\begin{definition}
    A map $F \, : \X \rightarrow \R^d$ is
    \begin{enumerate}
        \item \emph{Monotone} if, for all $\x, \y \in \X$,
        \begin{equation*}
            \langle F(\x) - F(\y), \x - \y \rangle \geq 0.
        \end{equation*}
        \item \emph{Strongly Monotone} with constant $\alpha > 0$ if, for all $\x, \y \in \X$,
        \begin{equation*}
            \langle F(\x) - F(\y), \x - \y \rangle \geq \alpha ||\x - \y||^2_2.
        \end{equation*}
    \end{enumerate}
\end{definition}
\begin{definition}[Monotone Game]
    A game $\game$ is \emph{monotone} if its pseudo-gradient map is monotone.
\end{definition}
A large part of our analysis will be in determining conditions under which the pseudo-gradient map is monotone. Upon doing so, we are able to employ the following results.
\begin{lemma}[\cite{melo:qre}] \label{lem::unique-qre}
    Consider a game $\game = (\agentset, \edgeset, (u_k, \actionset{k})_{k \in
    \agentset})$ and for any $T_1, \ldots, T_N > 0$, let $F$ be the pseudo-gradient map of $\game^H$. $\game$ has a unique QRE $\NE \in \Delta$ if $F$ is strongly monotone with any $\alpha>0$.
\end{lemma}
\begin{lemma}[\cite{hussain:aamas}] \label{lem::ql-conv}
    If the game $G$ is \emph{monotone}, then the Q-Learning Dynamics (\ref{eqn::QLD}) converge to the unique QRE with any positive exploration rates $T_1, \ldots, T_N > 0$.
\end{lemma}

\section{Convergence of Q-Learning in Network Games}

In this section we determine a sufficient condition under which Q-Learning converges to a unique QRE, which is given in terms of the exploration rate and the network game structure. To do this, we determine a sufficient condition on exploration rates $T_k$ such that the perturbed game $\game^H$ is strongly monotone. We find that this condition is dependent on the strength of pairwise interactions in the network, as well as its structure. We then compare our result to that of \cite{hussain:aamas} and show that, under suitable network structures, stability can be achieved with comparatively low exploration rates, even in the presence of many players. This also refines the result of \cite{sanders:chaos} which suggests that learning dynamics are increasingly unstable as the number of players increases, regardless of exploration rate.

To achieve our main result, we first parameterise pairwise interactions in a network game as follows.

\begin{definition}[Interaction Coefficient]
    Let $\game = (\agentset, \edgeset, (u_k, \actionset{k})_{k \in
    \agentset})$ be a network game whose edgeset is associated with the payoff functions $(A^{kl}, A^{lk})_{(k, l) \in \edgeset}$. Then, the \emph{interaction coefficient} $\delta_S$ of $\game$ is given as
    \begin{equation}
        \delta_S = \max_{(k, l) \in \edgeset} \lVert A^{kl} + (A^{lk})^\top \rVert_2,
    \end{equation}
    where $\lVert M \rVert_2 = \sup_{||\x||_2 = 1} \lVert M\x \rVert_2$ denotes the operator $2$-norm \cite{meiss:book}.
\end{definition}

\begin{theorem} \label{thm::main-thm}
    Consider a network game $\game = (\agentset, \edgeset, (u_k, \actionset{k})_{k \in
    \agentset})$ which has interaction coefficient $\delta_S$ and adjacency matrix $G$.
    The Q-Learning Dynamic converges to a unique QRE $\NE \in \Delta$ if, for all agents $k \in \agentset$, 
    \begin{equation} \label{eqn::main-cond}
        T_k > \frac{1}{2} \delta_S \left\lVert G \right\rVert_\infty,
    \end{equation}
    where $\lVert M \rVert_\infty = \max_i \sum_{j} |[G_{ij}]|$ is the operator $\infty$-norm.
\end{theorem}

We defer the full proof of Theorem \ref{thm::main-thm} to the Appendix and illustrate the main ideas here. In order to apply Lemma \ref{lem::ql-conv}, we must show that under (\ref{eqn::main-cond}), the perturbed game $\game^H$ is monotone. To do this, we decompose $\game^H$ into a term which is solely parameterised by exploration rates, and another term corresponding to the payoff matrices and graph structure. We then show that the second term can be decomposed as $\frac{1}{2}\delta_S\left\lVert G \right\rVert_\infty$, which allows us to separate terms involving the payoffs and the graph structure. We use the fact that the transformation between a $\game$ and $\game^H$ is given by $T_k \langle \x_k, \ln \x_{k} \rangle$, which has a strongly monotone gradient \cite{melo:qre} with constant $T_k$. Then, if the exploration rates are high enough to offset $\frac{1}{2}\delta_S \left\lVert G \right\rVert_\infty$, the resulting pseudo-gradient is monotone, and Lemma \ref{lem::ql-conv} can be applied.

The condition of Theorem \ref{thm::main-thm} for the convergence asserts that Q-Learning dynamics is convergent in a network game given sufficient exploration. In a similar light to the result of \cite{hussain:aamas}, the amount of exploration required depends on the strength of interaction. The main difference is that the condition includes a term $\left\lVert G \right\rVert_\infty$ which encodes the network structure. This term has a natural interpretation as follows. Let $\agentset_k = \{l \in \agentset : (k, l) \in \edgeset\}$ be the \emph{neighbours} of agent $k$, i.e. all the agents who interact with agent $k$ according to the network. Then $\lVert G \rVert_\infty = \max_k |\agentset_k|$, which denotes the maximum number of neighbours across all agents. 
\begin{figure*}[t]
\captionsetup{justification=centering}
	\centering
	\begin{subfigure}[b]{0.3\textwidth}
		\centering
		\includegraphics[width=0.8\textwidth]{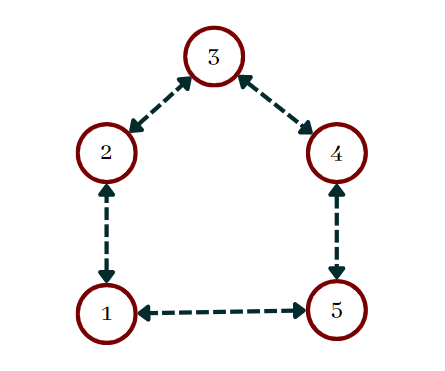}
		\caption*{Ring Network \\ $\left\lVert G \right\rVert_\infty = 2, \left\lVert G \right\rVert_2 = 2$}
	\end{subfigure}
	\begin{subfigure}[b]{0.3\textwidth}
		\centering
		\includegraphics[width=0.8\textwidth]{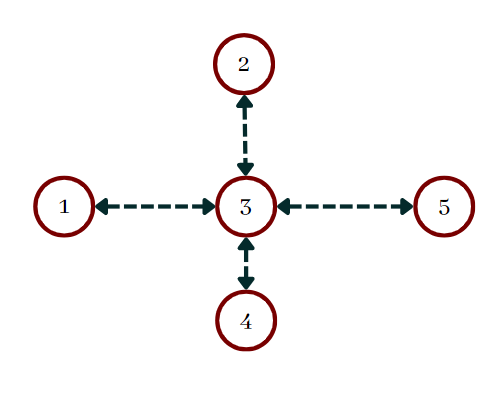}
		\caption*{Star Network \\ $\left\lVert G \right\rVert_\infty = N-1, \left\lVert G \right\rVert_2 = \sqrt{N - 1}$}
	\end{subfigure}
	\begin{subfigure}[b]{0.3\textwidth}
		\centering
		\includegraphics[width=0.8\textwidth]{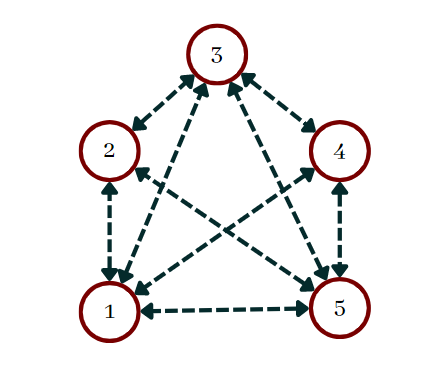}
		\caption*{Fully Connected Network \\ $\left\lVert G \right\rVert_\infty = N-1, \left\lVert G \right\rVert_2 = N - 1$}
	\end{subfigure}

 \caption{Examples of networks with five agents and associated $\left\lVert G \right\rVert_\infty$ and $\left\lVert G \right\rVert_2$. }\label{fig::example-networks}
\end{figure*}

\begin{figure*}[t]
	\centering
	\begin{subfigure}[b]{0.45\textwidth}
		\centering
		\includegraphics[width=1\textwidth]{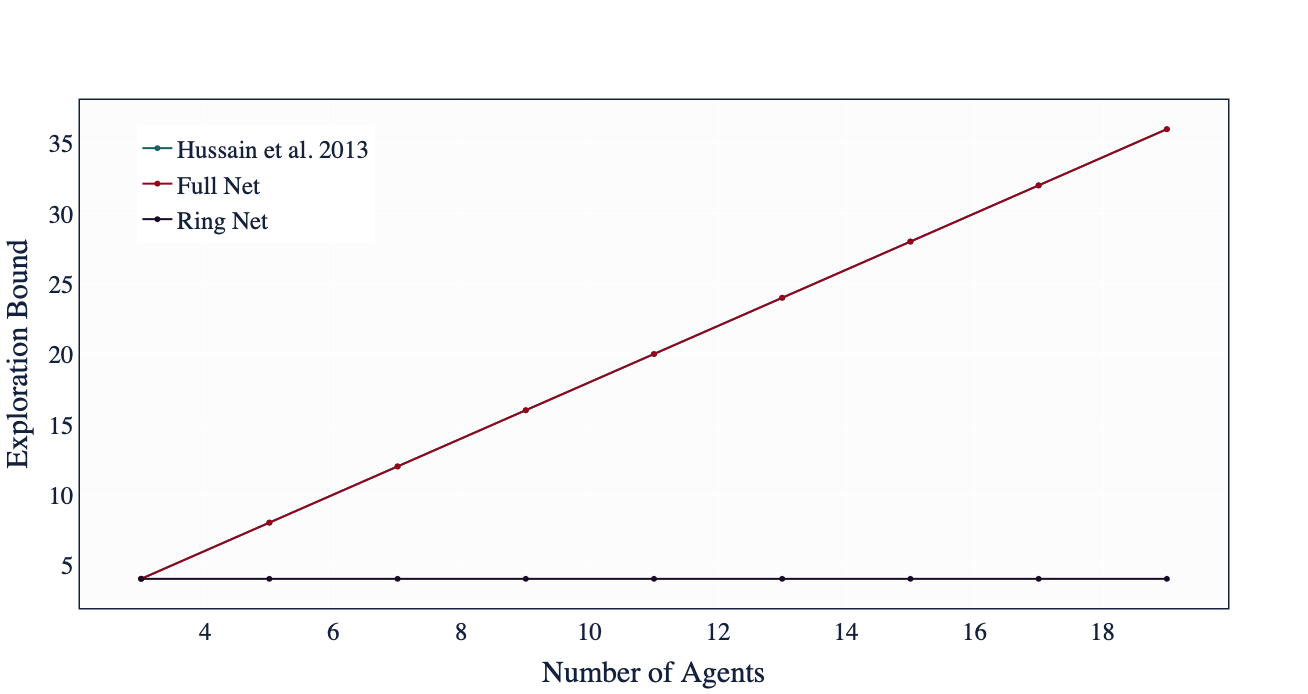}
		\caption*{Mismatching Game \\ $M = 2$\\ $\delta_S = 2$}
	\end{subfigure}
	\begin{subfigure}[b]{0.45\textwidth}
		\centering
		\includegraphics[width=1\textwidth]{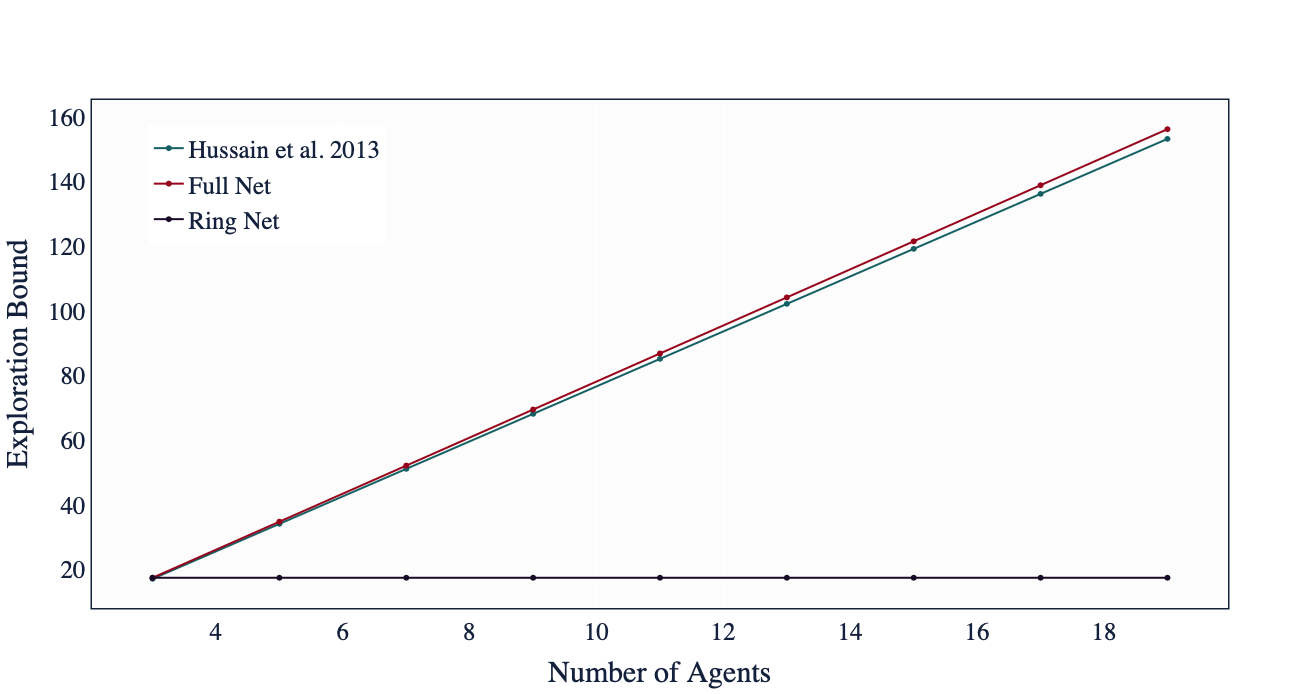}
		\caption*{Chakraborty Game \\ $\alpha = 7, \beta = 8.5$\\ $\delta_S \approx 8.67$}
	\end{subfigure}
	\begin{subfigure}[b]{0.45\textwidth}
		\centering
		\includegraphics[width=1\textwidth]{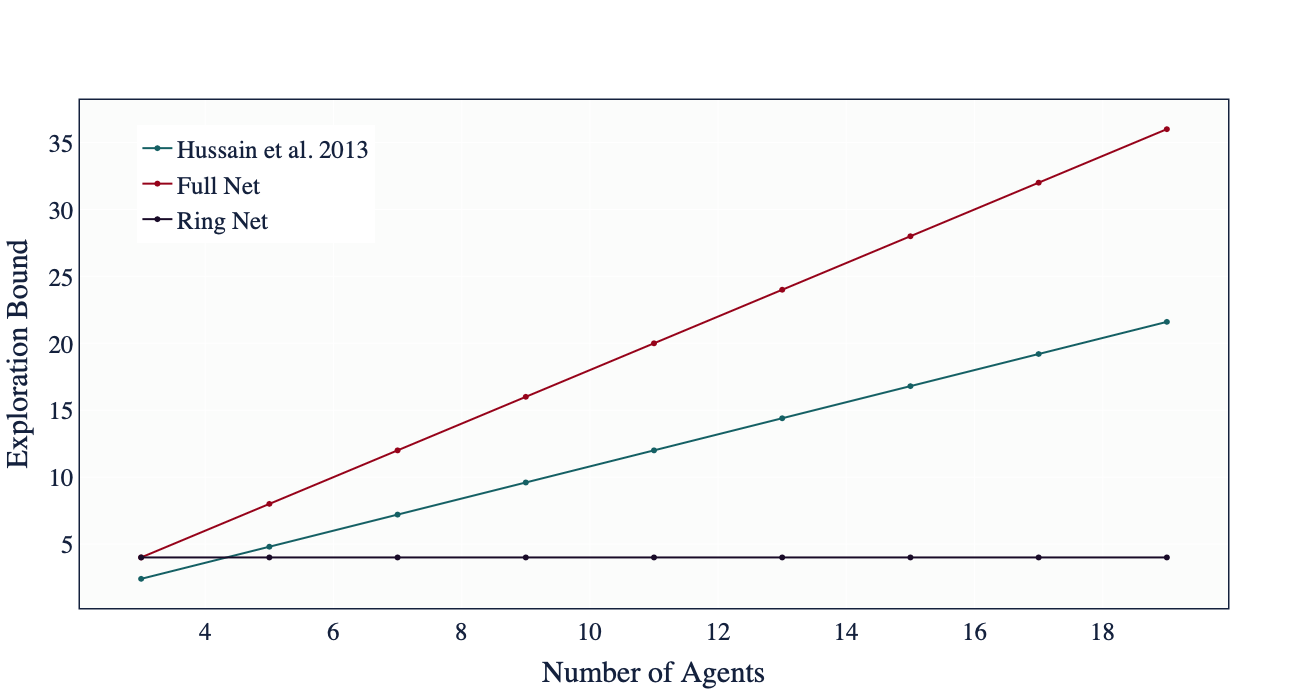}
		\caption*{Shapley Game \\ $\beta = 0.2$ \\ $\delta_S = 2$}
	\end{subfigure}
	\begin{subfigure}[b]{0.45\textwidth}
		\centering
		\includegraphics[width=1\textwidth]{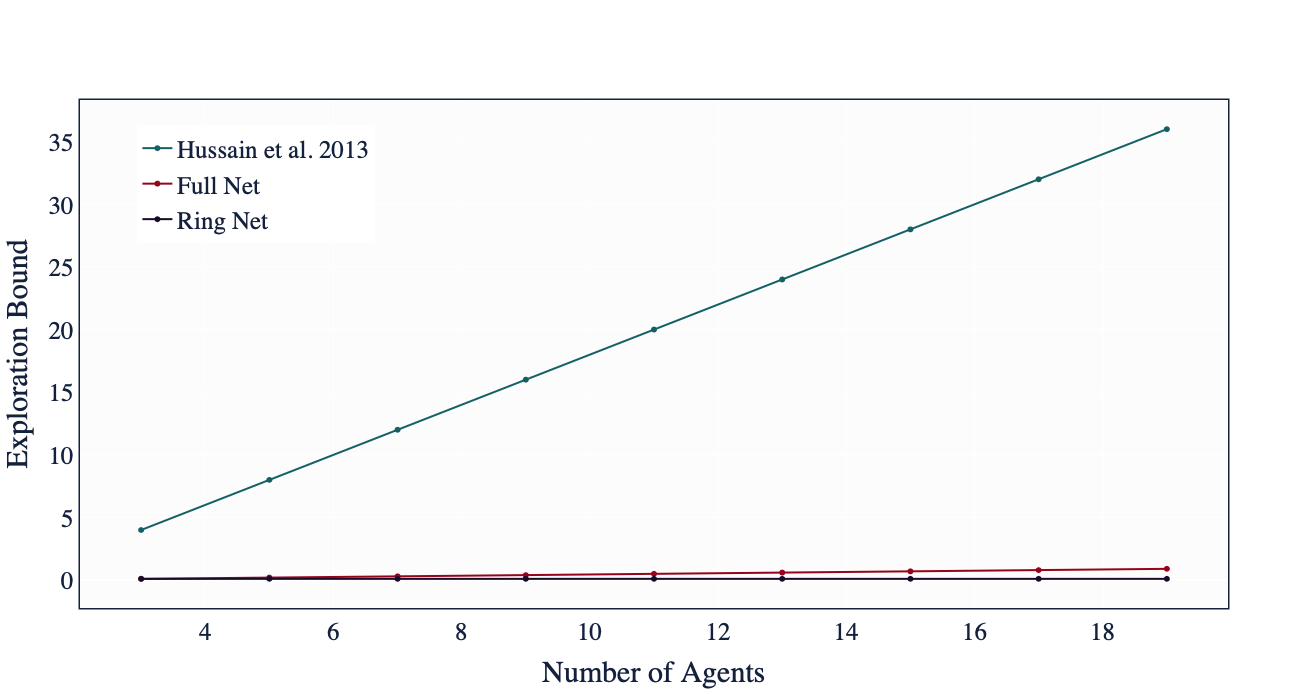}
		\caption*{Sato Game \\ $\epsX = 0.1, \epsY=-0.05$\\ $\delta_S = 0.05$}
	\end{subfigure}
 \caption{Lower Bound on sufficient exploration as defined by \cite{hussain:aamas} and by (\ref{eqn::main-cond}) in a fully connected network and ring network. The star network is not plotted as it has the same $\infty$-norm as the fully connected network. These are repeated for various games who are defined in Section \ref{sec::experiments}.}\label{fig::stability-boundary}
\end{figure*}
A useful point about (\ref{eqn::main-cond}) is that it does not make any assumptions regarding the nature of the interaction between games, but rather parameterises pairwise interactions by $\delta_S$. As such, the result is not limited to restrictive settings such as \emph{network zero sum games} \cite{piliouras:zerosum}. In fact, the convergence of Q-Learning dynamics in pairwise zero-sum games follows immediately from Theorem \ref{thm::main-thm}.

\begin{corollary}
    If the network game $\game$ is a pairwise zero-sum matrix, i.e., $A^{kl} + (A^{lk})^\top = 0$ for all $(k, l) \in \edgeset$, then the Q-Learning dynamics converge to a unique QRE so long as exploration rates $T_k$ for all agents are strictly positive.
\end{corollary}

\begin{remark}
    Corollary 1 is supported by the result of \cite{piliouras:zerosum,hussain:aamas} in which it was shown that Q-Learning converges to a unique QRE in all network zero-sum games (even if they are not pairwise zero-sum) so long as all exploration rates $T_k$ are positive.
\end{remark}


\paragraph{Discussion} The main takeaway from Theorem \ref{thm::main-thm} is that the condition on sufficient exploration depends on $\left\lVert G \right\rVert_\infty$, which is a measure of the network structure. In certain networks, such as the ring network depicted in Figure \ref{fig::example-networks}a, $\left\lVert G \right\rVert_\infty$ is independent of the number of agents in the system. Therefore, as the number of agents increase, the bound (\ref{eqn::main-cond}) does not increase. By contrast, in the fully connected network all agents are connected to each other and so $\left\lVert G \right\rVert_\infty$ increases with the number of agents. This illustrates the main point that the variation in the stability boundary defined by (\ref{eqn::main-cond}) depends on the structure of the network rather than solely on the total number of agents as previously found by \cite{hussain:aamas,sanders:chaos}. We illustrate this further in Figure \ref{fig::stability-boundary} which plots the stability boundary defined in \cite{hussain:aamas} in various games (which we define in Section \ref{sec::experiments}) as well as (\ref{eqn::main-cond}) for the ring and fully-connected network. Here, it is clear that (\ref{eqn::main-cond}) is a tighter bound than that of \cite{hussain:aamas} particularly for the ring network in all games. The advantage of using (\ref{eqn::main-cond}) is most clear in the example of the Sato game which, in \cite{sato:rps} was shown to display chaotic behaviour in the two-agent case when exploration rates are uniformly zero. In Figure \ref{fig::stability-boundary} it can be seen that only a small amount of exploration is required to stabilise the system.

\section{Experiments} \label{sec::experiments}
\begin{figure*}[t]
	\centering
	\begin{subfigure}[b]{0.225\textwidth}
		\centering
		\includegraphics[width=0.9\textwidth]{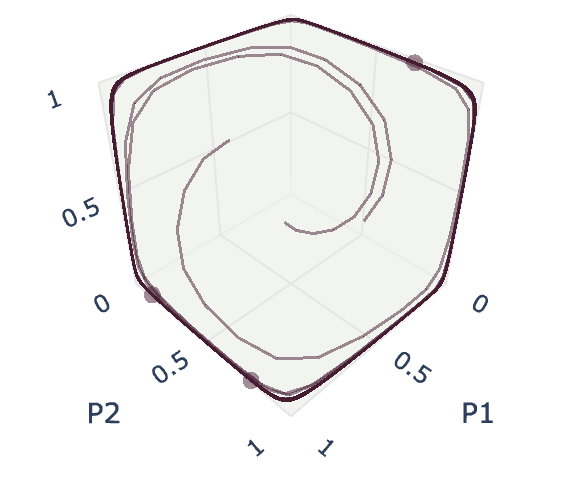}
		\caption*{$T = 0.7$}
	\end{subfigure}
	\begin{subfigure}[b]{0.225\textwidth}
		\centering
		\includegraphics[width=0.9\textwidth]{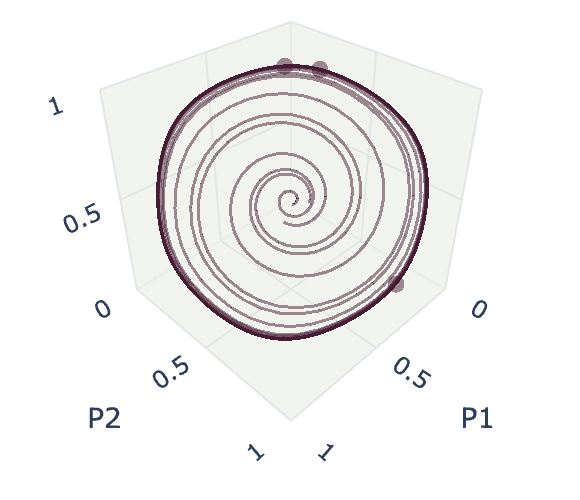}
		\caption*{$T = 1.5$}
	\end{subfigure}
	\begin{subfigure}[b]{0.225\textwidth}
		\centering
		\includegraphics[width=0.9\textwidth]{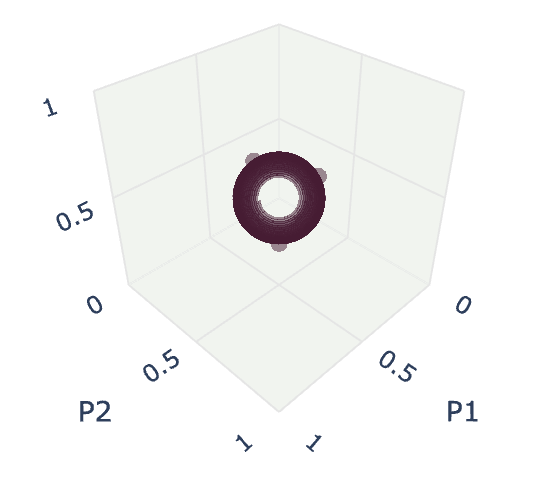}
		\caption*{$T = 2$}
	\end{subfigure}
	\begin{subfigure}[b]{0.225\textwidth}
		\centering
		\includegraphics[width=0.9\textwidth]{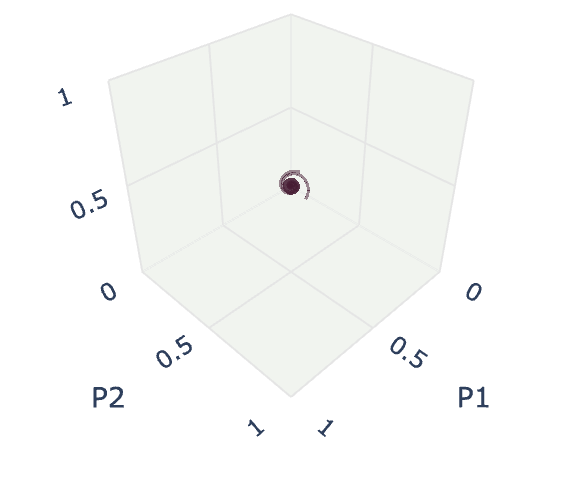}
		\caption*{$T = 2.7$}
	\end{subfigure}

 	\begin{subfigure}[b]{0.225\textwidth}
		\centering
		\includegraphics[width=0.9\textwidth]{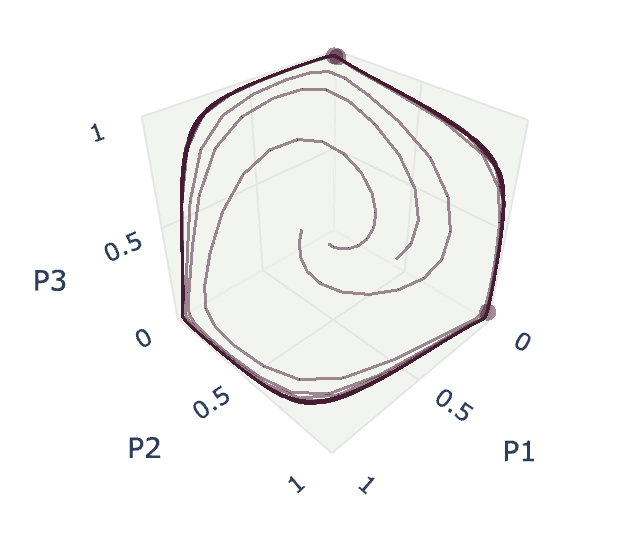}
		\caption*{$T = 0.15$}
	\end{subfigure}
	\begin{subfigure}[b]{0.225\textwidth}
		\centering
		\includegraphics[width=0.9\textwidth]{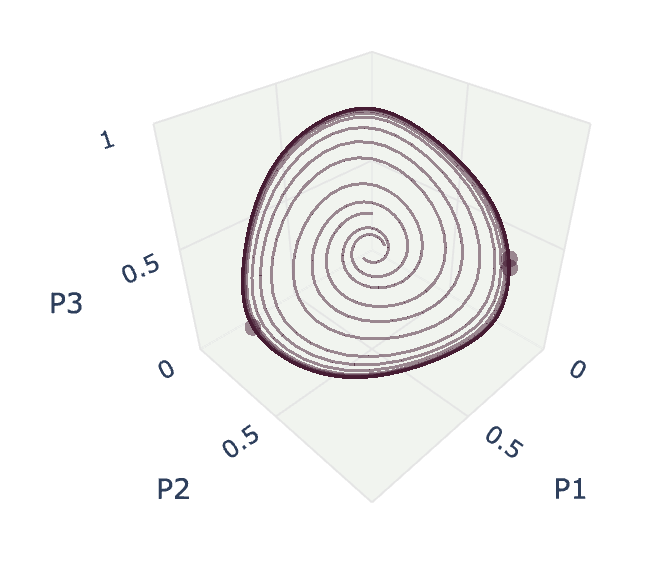}
		\caption*{$T = 0.3$}
	\end{subfigure}
	\begin{subfigure}[b]{0.225\textwidth}
		\centering
		\includegraphics[width=0.9\textwidth]{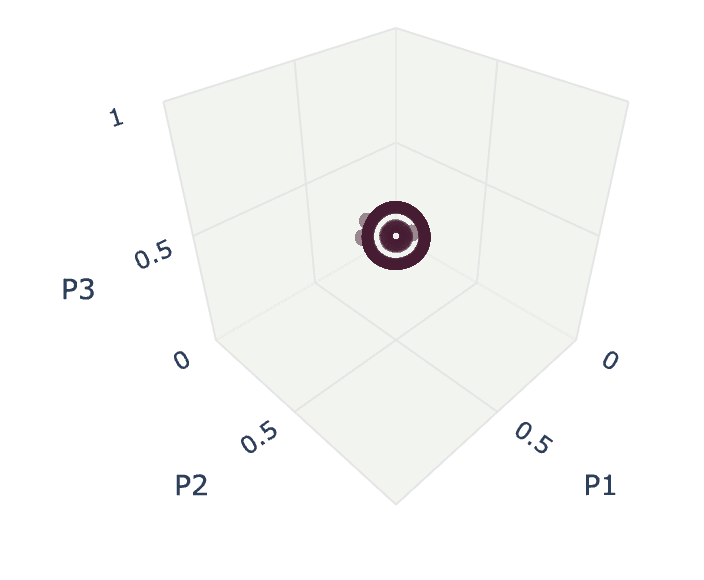}
		\caption*{$T = 0.4$}
	\end{subfigure}
	\begin{subfigure}[b]{0.225\textwidth}
		\centering
		\includegraphics[width=0.9\textwidth]{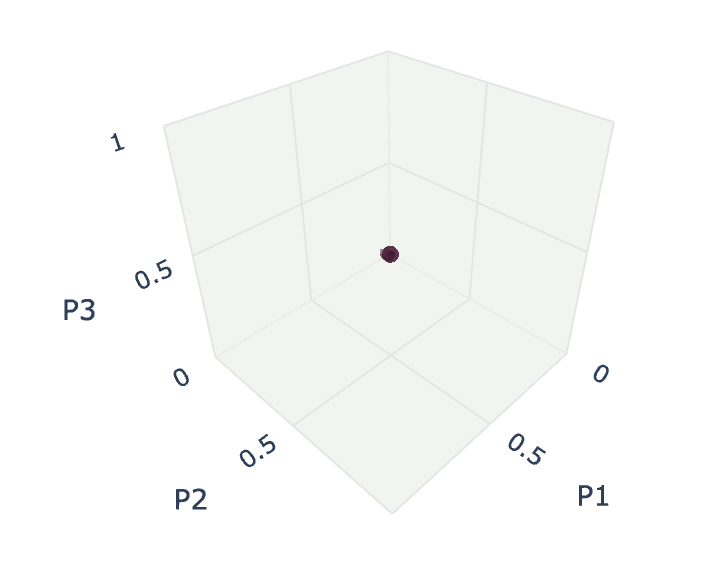}
		\caption*{$T = 0.55$}
	\end{subfigure}
 \caption{Trajectories of Q-Learning in a three agent (Top) Network Chakraborty Game with $\alpha = 7, \beta = 8.5$ (Bottom) Mismatching Game with $M = 2$. Axes denote the probabilities with which each player chooses their first action.}\label{fig::chakraborty-traj}
\end{figure*}

\begin{figure*}[t!]
	\centering
	\begin{subfigure}[b]{0.45\textwidth}
		\centering
		\includegraphics[width=1\textwidth]{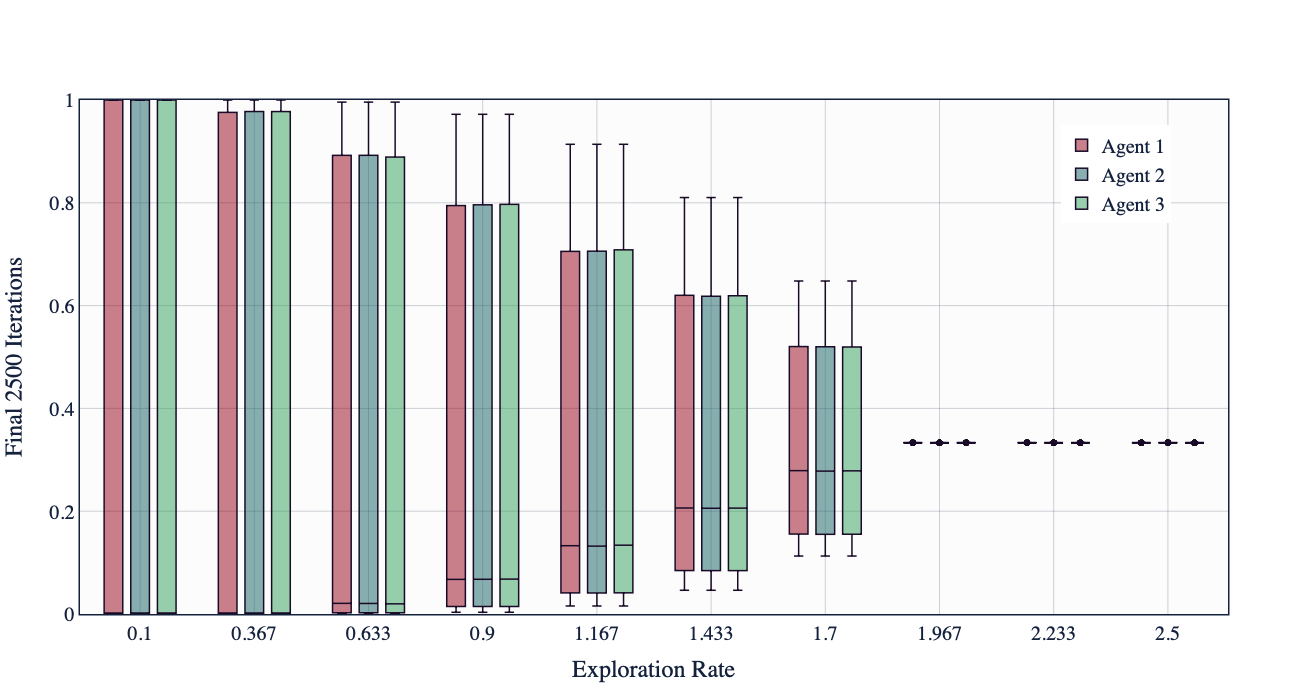}
		\caption*{Fully Connected Network}
	\end{subfigure}
	\begin{subfigure}[b]{0.45\textwidth}
		\centering
		\includegraphics[width=1\textwidth]{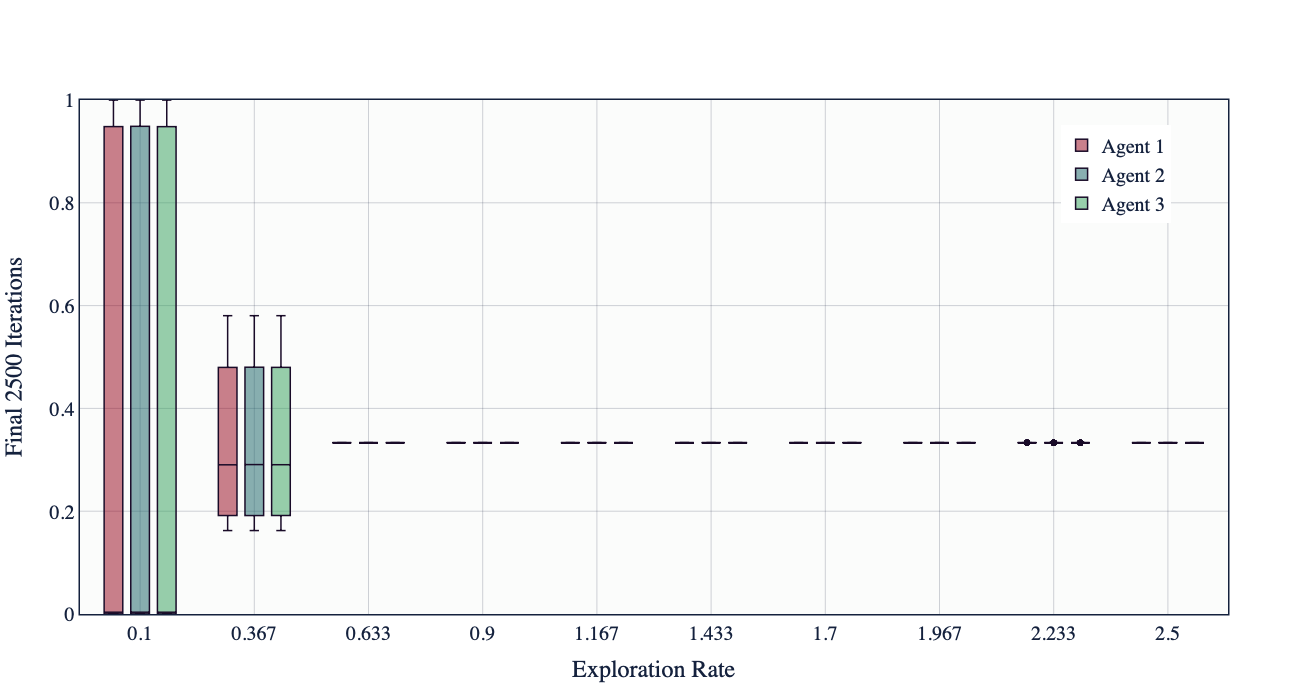}
		\caption*{Ring Network}
	\end{subfigure}
 	\begin{subfigure}[b]{0.45\textwidth}
		\centering
		\includegraphics[width=1\textwidth]{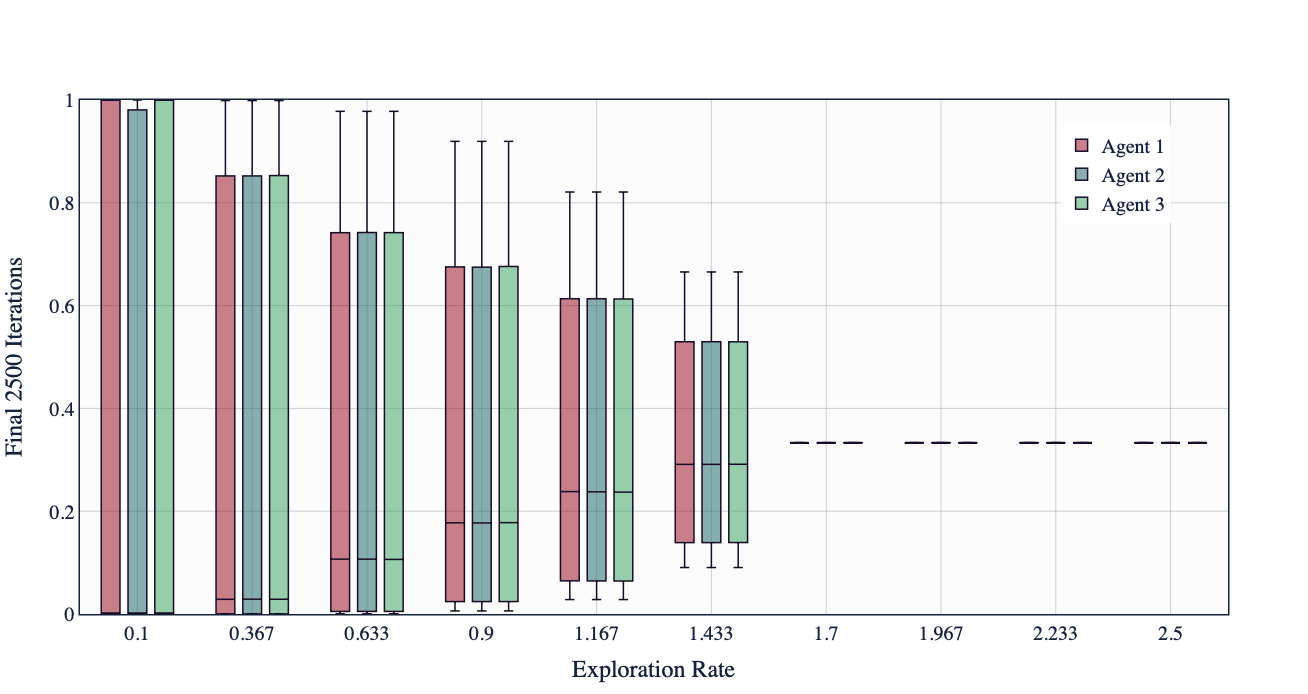}
		\caption*{Fully Connected Network}
	\end{subfigure}
	\begin{subfigure}[b]{0.45\textwidth}
		\centering
		\includegraphics[width=1\textwidth]{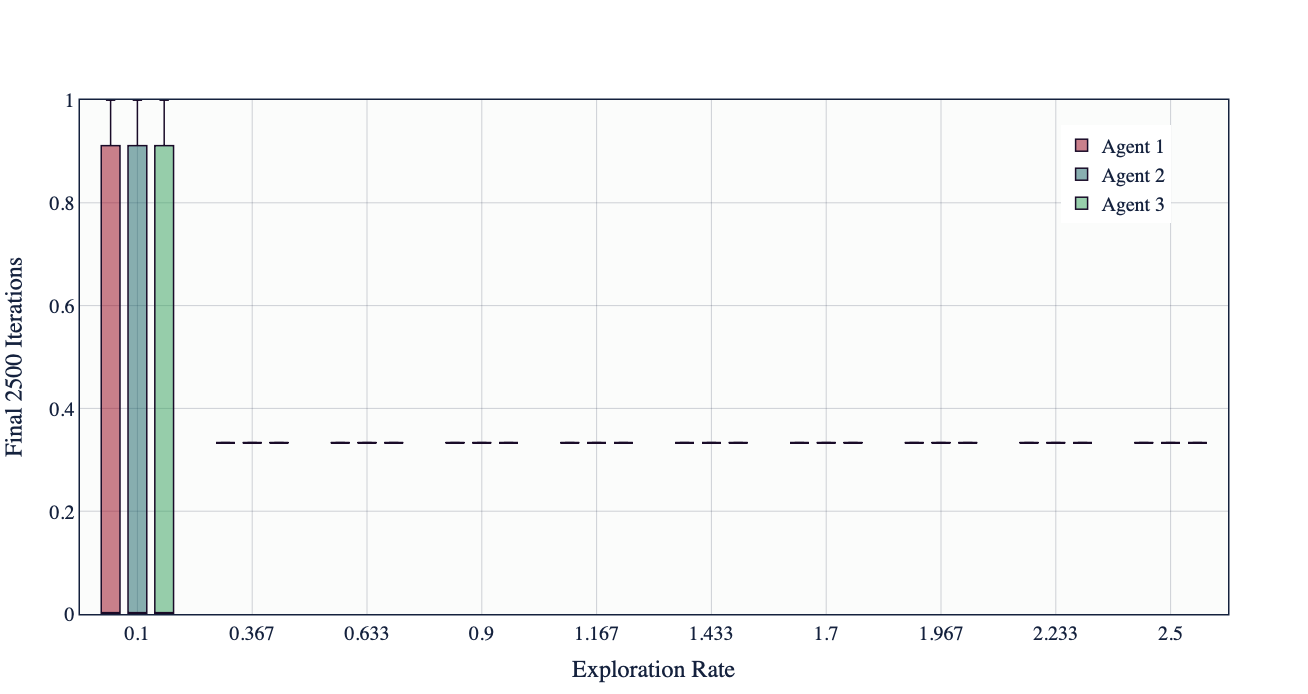}
		\caption*{Ring Network}
	\end{subfigure}
 \caption{Q-Learning in the (Top) Network Shapley Game (Bottom) Network Sato Game with 15 agents. The boxplot depicts the probabilities with which three of the agents play their first action in the final 2500 iterations of learning. This is depicted for varying choices of exploration rate $T$} \label{fig::shapley-box}
\end{figure*}
\begin{figure*}[t!]
	\centering
	\begin{subfigure}[b]{0.45\textwidth}
		\centering
		\includegraphics[width=1.1\textwidth]{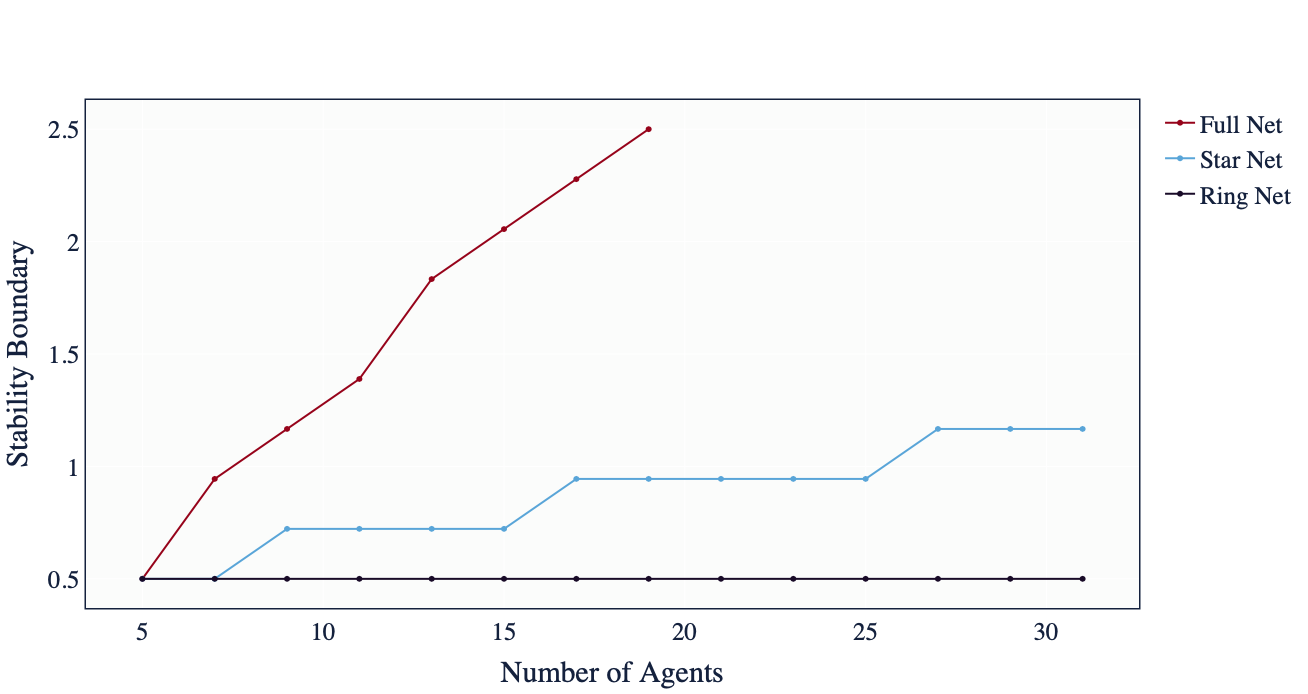}
		\caption*{Shapley Game}
	\end{subfigure}
	\begin{subfigure}[b]{0.45\textwidth}
		\centering
		\includegraphics[width=1.1\textwidth]{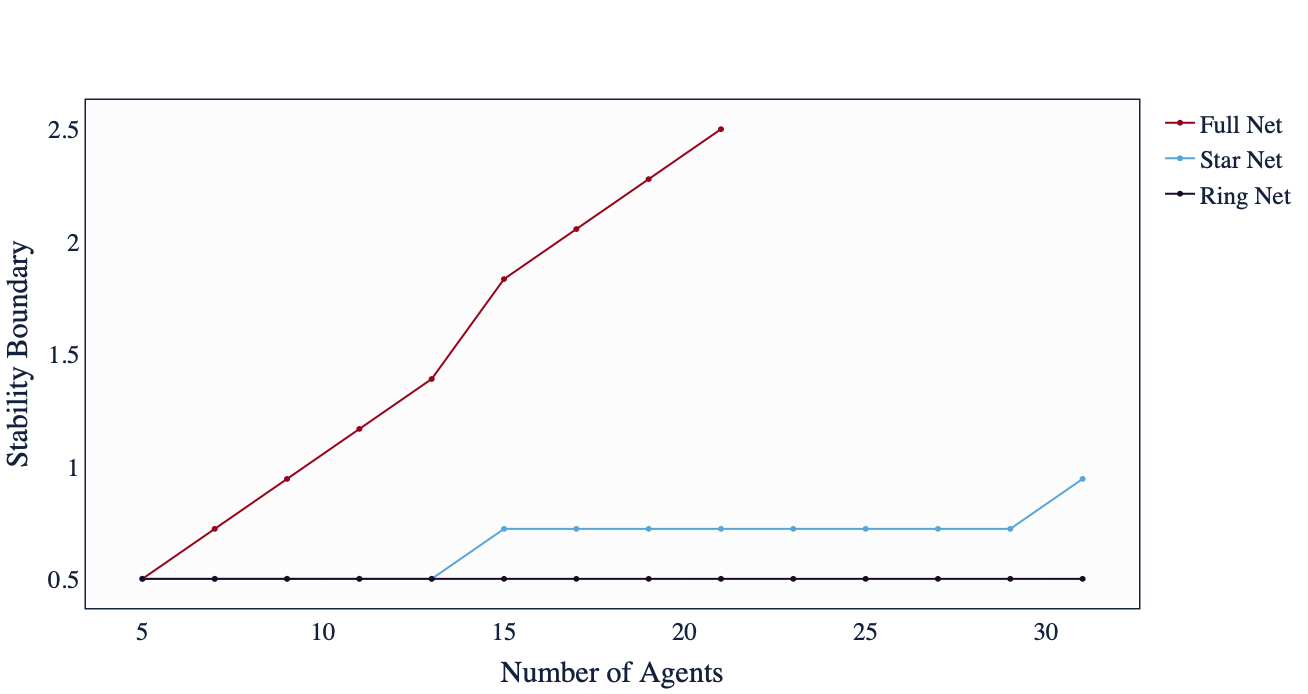}
		\caption*{Sato Game}
	\end{subfigure}
 \caption{Empirically determined stability boundary of Q-Learning measured against the number of agents. Q-Learning is iterated with 10 initial conditions and the game is considered to have converged if, for all agents and initial conditions (\ref{eqn::conv-criteria}) holds with $l = 1 \times 10^{-5}$. The Fully Connected Network, Star Network and Ring Networks are considered.} \label{fig::empirical-boundary}
\end{figure*}
In our experiments, we visualise and exemplify the implications of Theorem \ref{thm::main-thm} on a number of games. In particular, we simulate the Q-Learning algorithm described in Section \ref{sec::LearningModel} and show that Q-Learning asymptotically approaches a unique QRE so long as the exploration rates are sufficiently large. We show, in particular, that the amount of exploration required depends on the structure of the network rather than the total number of agents.

\begin{remark}
    In our experiments, we take all agents $k$ to have the same exploration rate $T$ and so drop the $k$ notation. As the bound (\ref{eqn::main-cond}) must hold for all agents $k$, this assumption does not affect the generality of the results.
\end{remark}

\paragraph{Convergence of Q-Learning.} We first illustrate the convergence of Q-Learning using two representative examples: the \emph{Network Chakraborty Game} and the \emph{Mismatching Pennies Game}. The former was first analysed in \cite{chakraborty:deviations} to characterise chaos in learning dynamics. Formally, the payoff to each agent $k$ is defined as
\begin{align*}
    &u_k(\x_k, \x_{-k}) = \x_k^\top \A \x_l, \; l = k-1 \mod N, \\
    &A = \begin{pmatrix}
        1 & \alpha \\ \beta & 0
    \end{pmatrix}, \; \alpha, \beta \in \R.
\end{align*}
The latter was first analysed in \cite{kleinberg:nashbarrier} in which it was shown that learning dynamics reach a cycle around the boundary of the simplex. Here, the payoffs to each agent are given by
\begin{align*}
    &u_k(\x_k, \x_{-k}) = \x_k^\top \A \x_l, \; l = k-1 \mod N, \\
    &A = \begin{pmatrix}
        0 & 1 \\ M & 0
    \end{pmatrix}, \; M \geq 1.
\end{align*}
We visualise the trajectories generated by running Q-Learning in Figure \ref{fig::chakraborty-traj} in both games for a three agent network and choosing $\alpha = 7, \beta = 8.5, M = 2$. It can be seen that, for low exploration rates, the dynamics reach a limit cycle around the boundary of the simplex. However, as exploration increases, the dynamics are eventually driven towards a fixed point for all initial conditions. The higher requirement on exploration in the Chakraborty Game as compared to the Mismatching Game can be seen as stemming from the higher $\delta_S \approx 8.67$ in the former compared to $\delta_S = 2$ in the latter. 

\paragraph{Network Shapley Game} In the following example, each edge of the network game has associated the same pair of matrices $A, B$ where
\begin{eqnarray*} 
A=\begin{pmatrix}
1 & 0 & \beta \\
            \beta & 1 & 0 \\ 0 & \beta & 1
            \end{pmatrix}, \, B=\begin{pmatrix}
            -\beta & 1 & 0 \\ 0 & -\beta & 1 \\ 1 & 0 & -\beta
            \end{pmatrix},
\end{eqnarray*}    
where $\beta \in (0, 1)$.

This has been analysed in the two-agent case in \cite{shapley:twoperson}, where it was shown that the \emph{Fictitious Play} learning dynamic do not converge to an equilibrium. \cite{hussain:aamas} analysed the network variant of this game for the case of a ring network and numerically showed that convergence can be achieved by Q-Learning through sufficient exploration. In Figure \ref{fig::shapley-box} we examine both a fully connected network and a ring network with 15 agents. Figure \ref{fig::shapley-box} depicts the final 2500 iterations of learning for three agents and 35 initial conditions. It can be seen that, as exploration rates increase Q-Learning is driven towards an equilibrium for all initial conditions. Importantly, the boundary at which equilibrium behaviour occurs is higher in the fully connected network, where $\left\lVert G \right\rVert_\infty = 14$ than in the ring network, where $\left\lVert G \right\rVert_\infty = 2$.

\paragraph{Network Sato Game} We also analyse the behaviour of Q-Learning in a variant of the game introduced in \cite{sato:rps}, where it was shown that chaotic behaviour is exhibited by learning dynamics in the two-agent case. We extend this towards a network game by associating each edge with the payoff matrices $A, B$ given by
\begin{eqnarray*} 
A=\begin{pmatrix}
\epsX & -1 & 1 \\
1 & \epsX & -1 \\
-1 & 1 & \epsX
\end{pmatrix}, \, B=\begin{pmatrix}
\epsY & -1 & 1 \\
1 & \epsY & -1 \\
-1 & 1 & \epsY
\end{pmatrix},
\end{eqnarray*}   
where $\epsX, \epsY \in \R$. Notice that for $\epsX = \epsY = 0$, this corresponds to the classic Rock-Paper-Scissors game which is zero-sum so that, by Corollary 1, Q-Learning will converge to an equilibrium with any positive exploration rates. We choose $\epsX=0.01, \epsY=-0.05$ in order to stay consistent with \cite{sato:rps} which showed chaotic dynamics for this choice. The boxplot once again shows that sufficient exploration leads to convergence of all initial conditions. However, the amount of exploration required is significantly smaller than that of the Network Shapley Game. This can be seen as being due to the significantly lower interaction coefficient of the Sato game $\delta_S = 0.05$ as compared to the Shapley game $\delta_S = 2$.

\paragraph{Stability Boundary} In these experiments we empirically determine the dependence of the stability boundary w.r.t.~the number of agents. For accurate comparison with Figure \ref{fig::stability-boundary}, we consider the Network Sato and Shapley Games in a fully-connected network, star network and ring network. We iterate Q-Learning for various values of $T$ and determine whether the dynamics have converged. To evaluate convergence, we record the final 2500 iterations and check whether the relative difference between the maximum and minimum strategy components $x_{ki}$ is less than some tolerance $l$ for all agents $k$, actions $i$ and initial conditions. More formally we aim to determine if
\begin{equation}\label{eqn::conv-criteria}
    \lim_{t \rightarrow \infty} \left( \frac{\max_t x_{ki}(t) - \min_t x_{ki}(t)}{\max_t x_{ki}(t)} \right) < l
\end{equation}
holds for all $k \in \agentset$ and all $i \in \actionset{k}$. In Figure \ref{fig::empirical-boundary} we plot the smallest exploration rate $T$ for which (\ref{eqn::conv-criteria}) holds for varying choices of $N$, using $l = 1\times10^{-5}$. It can be seen that the prediction of (\ref{eqn::main-cond}) holds, in that the number of agents plays no impact for the ring network whereas the increase in the fully-connected network is linear in $N$. In addition, it is clear that the stability boundary increases slower in the Sato game than in the Shapley game, owing to the smaller interaction coefficient. 

An additional point to note is that the stability boundary for the star network increases slower than the fully-connected network in all games. We anticipate that this is due to the fact that the $2$-norm $\lVert G \rVert_2$ in the star network is smaller than that of the fully-connected network (c.f.¬Figure \ref{fig::example-networks}). We therefore conjecture that a tighter lower bound on exploration can be obtained using the $2$-norm, which we consider an important avenue for future work.

\section{Conclusion}

In this paper we show that the Q-Learning dynamics is guaranteed to converge in arbitrary network games, independent of any restrictive assumptions such as network zero-sum or potential. This allows us to make a branching statement which applies across all network games.

In particular, our analysis shows that convergence of the Q-Learning dynamics can be achieved through sufficient exploration, where the bound depends on the pairwise interaction between agents and the structure of the network. Overall, compared to the literature, we are able to tighten the bound on sufficient exploration and show that, under certain network interactions, the bound does not increase with the total number of agents. This allows for stability to be guaranteed in network games with many players.

A fruitful direction for future research would be to capture the effect of the payoffs through a tighter bound than the interaction coefficient and to explore further how properties of the network affect the bound. In addition, whilst there is still much to learn in the behaviour of Q-Learning in stateless games, the introduction of the state variable in the Q-update is a valuable next step.

\section*{Acknowledgements}

Aamal Hussain and Francesco Belardinelli are partly funded by the UKRI Centre for Doctoral Training in Safe and Trusted Artificial Intelligence (grant number EP/S023356/1). Dan Leonte acknowledges support from the EPSRC Centre for Doctoral Training in Mathematics of Random Systems: Analysis, Modelling and Simulation (EP/S023925/1). This research/project is supported in part by the National Research Foundation, Singapore and DSO National Laboratories under its AI Singapore Program (AISG Award No: AISG2-RP-2020-016), NRF 2018 Fellowship NRF-NRFF2018-07, NRF2019-NRF-ANR095 ALIAS grant, grant PIESGP-AI-2020-01, AME Programmatic Fund (Grant No.A20H6b0151) from the Agency for Science, Technology and Research (A*STAR) and Provost’s Chair Professorship grant RGEPPV2101.


\bibliography{references}
\bibliographystyle{icml2023}

\newpage
\appendix
\onecolumn

\section{Proof of Theorem \ref{thm::main-thm}}

\paragraph{Preliminaries} We begin in this section by defining the various tools that we will use in our proof. Recall that an operator $f : \X \subset \R^n \rightarrow \R^n$ is \emph{strongly convex} with constant $\alpha$ if, for all $\x, \y \in \X$
\begin{equation*}
    f(\y) \geq f(\x) + Df(\x)^\top (\y - \x) + \frac{\alpha}{2}\lVert\x - \y \rVert^2_2.
\end{equation*}
It is known that, if $f(\x)$ is strongly convex, then its Hessian $D^2_{\x} f(\x)$ is strongly positive definite with constant $\alpha$. Thus, all eigenvalues of $D^2_{\x} f(\x)$ are larger than $\alpha$. To apply this in our setting, we use the following result.
\begin{proposition}[\cite{melo:qre}] \label{prop::strong-conv}
        The function $f(\x_k) = T_k \langle \x_k, \ln \x_k \rangle$ is strongly convex with constant $T_k$.
\end{proposition}
Then, $D^2_{\x_k} f(\x_k)$ has eigenvalues larger than $T_k$. 

In addition, the following definitions and properties hold for any matrix $A$.
\begin{enumerate}
    \item $\lVert A \rVert_2 = \sqrt{\lambda_{\max}(A^\top A)}$ where $\lambda_{\max}$ is the largest eigenvalue of $A$
    \item $\lVert A \rVert_\infty = \max_i \sum_{j} |[A]_{ij}|$
    \item $\rho(A) = \max_i |\lambda_i(A)|$ where $\lambda_i(A)$ denotes an eigenvalue of $A$
\end{enumerate}
\begin{proposition}[Weyl's Inequality] \label{prop::weyl}
        Let $J = D + N$ where $D$ and $N$ are symmetric matrices. Then it holds that
        \begin{equation*}
            \lambda_{\min}(J) \geq \lambda_{\min} (D) + \lambda_{\min} (N).
        \end{equation*}
        where $\lambda_{\min}(A)$ denotes the smallest eigenvalue of a matrix. 
\end{proposition}
\begin{proposition} \label{prop::spectral-radius}
    Let $A$ be a symmetric matrix. Then
    \begin{equation*}
       |\lambda_{\min}(A)| \leq \rho(A) = \lVert A \rVert_2.
       \end{equation*}
\end{proposition}
The following result is used in our proof to be able to parameterise the effect of pairwise interactions by $\delta_S$.
\begin{lemma} \label{lem::two-norm}
    Let $G \in \mathcal{M}_{N}(\R)$ be matrix for which each entry $g_{ij} \defeq \left[G\right]_{ij}$ is either $0$ or $1$. Let $N \in \mathcal{M}_{Nn}(\R)$ be a block matrix such that
    \begin{equation*}
        \left[N\right]_{ij} = \begin{cases}
            A^{ij} & \text{ if } g_{ij} = 1 \\ 
            \zeros & \text{ otherwise}
        \end{cases},
    \end{equation*}
    where $A^{ij} \in \mathcal{M}_n(\mathbb{R})$ are matrices of the same dimension. 
    Then 
    \begin{equation*}
 \left\lVert N \right\rVert_2 \leq \sqrt{\left\lVert G \right\rVert_1 \left\lVert G \right\rVert_\infty}\max_{1 \le i,j \le n} \left\lVert A_{ij} \right\rVert_2.
    \end{equation*}
\end{lemma}
\begin{proof} Let $v = (v^1,\ldots,v^n)\in \mathbb{R}^{Nn}$ where $v^{i} \in \mathbb{R}^N$ for $1 \le i \le n$. Then
\begin{equation}\lVert Nv \rVert_2^2 =   \left\lVert\begin{pmatrix}
            g_{11}A^{11} & \ldots  & g_{1n}A^{1n} \\
            \vdots &  &   \vdots \\
             g_{n1}A^{n1} & \ldots & g_{nn}A^{nn} 
        \end{pmatrix} \begin{bmatrix}
            v^1\\ \vdots  \\ v^n
        \end{bmatrix}\right\rVert_2^2 =\left\lVert \begin{bmatrix}
          \sum_{1j} g_{1j} A^{1j} v^j \\ \vdots \\ \sum_{ni} g_{nj}A^{nj} v^j
    \end{bmatrix}\right\rVert_2^2 \leq \sum_{i=1}^n \left\lVert \sum_{j=1}^n g_{ij} A^{ij}v^j  \right\rVert_2^2.\label{eq:squared_norm_of_nv}
        \end{equation}
For each fixed $i \in \{1,\ldots,n\}$, we have the upper bound
\begin{equation}
\left\lVert \sum_{j=1}^n g_{ij} A^{ij}v^j  \right\rVert_2 \le \sum_{j=1}^n g_{ij} \left\lVert A^{ij}v^j  \right\rVert_2 \le \sum_{j=1}^n g_{ij} \left\lVert A^{ij}  \right\rVert_2 \left\lVert v^j \right\rVert_2 \le \max_{1 \le i,j \le n}\left\lVert A^{ij}  \right\rVert_2 \sum_{j=1}^n g_{ij}  \left\lVert v^j \right\rVert_2.\label{eq:bound_for_each_i}
\end{equation}
By plugging \eqref{eq:bound_for_each_i} in \eqref{eq:squared_norm_of_nv} and expanding the squared bracket, we obtain that
\begin{align*}
    \lVert Nv \rVert_2^2 \le \sum_{i=1}^n \left(\max_{1 \le i,j \le n}\left\lVert A^{ij}  \right\rVert_2 \sum_{j=1}^n g_{ij}  \left\lVert v^j \right\rVert_2 \right)^2 =& \max_{1 \le i,j \le n}\left\lVert A^{ij}  \right\rVert_2^2 \sum_{i=1}^n \sum_{k,l=1}^n g_{ik} g_{il} \left\lVert v^k \right\rVert_2 \left\lVert v^l \right\rVert_2\\ 
    \le & \max_{1 \le i,j \le n}\left\lVert A^{ij}  \right\rVert_2^2 \sum_{i=1}^n \sum_{k,l=1}^n g_{ik} g_{il} \left(\frac{1}{2}\left\lVert v^k \right\rVert_2^2 + \frac{1}{2}\left\lVert v^l \right\rVert_2^2 \right),
\end{align*}
where the last inequality follows by completing the square. Notice that the two sums above are identical, hence
\begin{equation*}
    \lVert Nv \rVert_2^2 \le \max_{1 \le i,j \le n}\left\lVert A^{ij}  \right\rVert_2^2 \sum_{i=1}^n \sum_{k,l=1}^n g_{ik} g_{il} \left\lVert v^k \right\rVert_2^2. 
\end{equation*}
It remains the upper bound the RHS in the above inequality. Indeed, we have that
\begin{align*}
    \sum_{i=1}^n \sum_{k,l=1}^n g_{ik} g_{il} \left\lVert v^k \right\rVert_2^2 &= \sum_{i=1}^n \sum_{k=1}^n g_{ik} \left\lVert v^k \right\rVert_2^2 \left(\sum_{l=1}^n  g_{il}\right) \le  \left\lVert G_\infty\right\rVert \sum_{i=1}^n \sum_{k=1}^n g_{ik} \left\lVert v^k \right\rVert_2^2 \\
    &\le \left\lVert G_\infty \right\rVert \sum_{k=1}^n \left(\sum_{i=1}^n g_{ik}\right) \left\lVert v^k \right\rVert_2^2 \le \left\lVert G_\infty \right\rVert \left\lVert G_1 \right\rVert \sum_{k=1}^n \left\lVert v^k \right\rVert_2^2 = \left\lVert G_\infty \right\rVert \left\lVert G_1 \right\rVert.
\end{align*}
Thus \begin{equation*}
\sup_{v \colon \left\lVert v\right\rVert =1} \left\lVert Nv\right\rVert_2^2\le \left\lVert G_\infty \right\rVert \left\lVert G_1 \right\rVert     \max_{i,j} \left\lVert A^{ij} \right\lVert_2^2,
\end{equation*}
and the conclusion follows.
\end{proof}
With these results in place, we can prove Theorem \ref{thm::main-thm} in the main paper.
\begin{proof}[Proof of Theorem \ref{thm::main-thm}]
    In order to apply Lemma \ref{lem::ql-conv} we show that, under the condition (\ref{eqn::main-cond}), the perturbed game $\game^H$ is strongly monotone. To this end, we take the derivative of the pseudo-gradient of $\game^H$ which we call the \emph{pseudo-Hessian} given by
    \begin{equation*}
        [J(\x)]_{k, l} = D_{\x_l} F_k(\x).
    \end{equation*}
    It follows that, if $\frac{J(\x) + J^\top(\x)}{2}$ is strongly positive definite for all $\x \in \Delta$ with any $\alpha > 0$, i.e.  $\x^\top J(\x) \x \geq \alpha$ for all $\x \in \Delta$, then $F(\x)$ is strongly monotone with the same constant $\alpha$. We can rewrite the pseudo-Hessian as
    \begin{equation*}
        J(\x) = D(\x) + N(\x),
    \end{equation*}
    where $D(\x)$ is a block diagonal matrix with $-D^2_{\x_k \x_k} u_k^H(\x_k, \x_{-k})$ along the diagonal. $N(\x)$ is an off-diagonal block matrix with
    \begin{equation*}
        [N(\x)]_{k, l} = \begin{cases}
            - D_{\x_k, \x_l} u_k^H(\x_k, \x_{-k}) &\text{ if } (k, l) \in \edgeset \\
            \zeros &\text{ otherwise}
        \end{cases}.
    \end{equation*}
    In words, $N(x)$ shares the same structure of the adjacency matrix $G$ of the game, except that it has $-D_{\x_k, \x_l} u_k^H(\x_k, \x_{-k})$ wherever $G$ takes the value $1$ and the block matrix $\zeros$ wherever $G$ has $0$.
    Next we evaluate these partial differentials. Recall that
    \begin{equation*}
        -u_k^H(\x_k, \x_{-k}) = T_k \langle \x_k, \ln \x_k \rangle - \sum_{(k, l) \in \edgeset} \x_k \cdot A^{kl} \x_l.
    \end{equation*}
    As a result, for all $(k, l) \in \edgeset$, $\left[N(\x)\right]_{k, l} = - A^{kl}$, so that $N(\x)$ represents the network interaction. By contrast, $D(\x)$ depends on $T_k$ and is independent of the payoffs $u_k$. As such, it measures the strength of the game perturbation. Now, let $\Bar{J}(\x)$ be defined as
    \begin{align*}
        \Bar{J}(\x) &= \frac{J(\x) + J^\top(\x)}{2}\\
                    &= D(\x) + \frac{N(\x) + N^\top(\x)}{2}.
    \end{align*}
    Then we apply the following results.

    Then, from Proposition \ref{prop::strong-conv} it  follows that $D(\x)$ is strongly positive definite with constant $T = \min_{k} T_k$. In particular, this means that $\lambda_{\min} D(\x) \geq T$. Finally, applying Weyl's inequality
    \begin{align*}
        \lambda_{\min}(\Bar{J}) &\geq T + \lambda_{\min} \left( \frac{N + N^\top}{2} \right) \\
        &\geq T - \rho\left( \frac{N + N^\top}{2} \right) \\
        &= T - \left\lVert\frac{N + N^\top}{2} \right\rVert_2 \\
        &\geq T - \frac{1}{2} \left\lVert A + B^\top \right\lVert_2  \sqrt{\lVert G \rVert_\infty \lVert G \rVert_1}  \\
        &= T - \frac{1}{2} \left\lVert A + B^\top \right\rVert_2 \left\lVert G \right\rVert_\infty\\
        &= T - \frac{1}{2} \delta_S \left\lVert G \right\rVert_\infty
    \end{align*}
    where we employ Propositions \ref{prop::spectral-radius}, Lemma \ref{lem::two-norm} and the fact that $G$ is symmetric so that $\lVert G \rVert_\infty = \lVert G \rVert_1$. The matrices $A, B$ are chosen so that
    \begin{equation*}
        \left\lVert A + B^\top \right\rVert_2 = \max_{(k, l) \in \edgeset} \left\lVert A^{kl} + (A^{lk})^\top \right\rVert_2 = \delta_S.
    \end{equation*}
    Then, under (\ref{eqn::main-cond}), $\lambda_{\min}(\Bar{J}(\x)) \geq T - \frac{1}{2} \delta_S \lVert G \rVert_\infty > 0$ and, therefore $F(\x)$ is strongly monotone with constant $T - \frac{1}{2} \delta_S \left\lVert G \right\rVert_\infty$. Using Lemma \ref{lem::ql-conv}, it follows that Q-Learning Dynamics converge to a unique QRE.
    \end{proof}

\end{document}